\pdfoutput=1

\documentclass[11pt,letterpaper]{article}
\usepackage{fullpage}

\renewcommand{\emph}[1]{\textit{\textbf{#1}}}

\newenvironment{proof}{\noindent{\bf Proof:}}{\hspace*{\fill}\rule{6pt}{6pt}\bigskip}
\newtheorem{theorem}{Theorem}[section]

\newcommand{\comment}[1]{}

\renewcommand{\Pr}{\mathop{\bf Pr}\nolimits}
\newcommand{\E}{\mathop{\bf E}\nolimits}
\long\def\skipthis#1{}

\begin{document}

\title{On the Zero-Error Capacity Threshold for Deletion Channels}

\author{Ian A. Kash \and
Michael Mitzenmacher \and
Justin Thaler \and
Jonathan Ullman\\
School of Engineering and Applied Sciences\\
Harvard University,
Cambridge, MA 02138\\ Email: \{kash, jthaler, jullman\}@seas.harvard.edu \, \, michaelm@eecs.harvard.edu} 
\date{}

\maketitle


\begin{abstract}
We consider the zero-error capacity of deletion channels.
Specifically, we consider the setting where we choose a codebook
${\cal C}$ consisting of strings of $n$ bits, and our model of the
channel corresponds to an adversary who may delete up to $pn$ of these
bits for a constant $p$.  Our goal is to decode correctly without
error regardless of the actions of the adversary.  We consider what
values of $p$ allow non-zero capacity in this setting.  We suggest
multiple approaches, one of which makes use of the natural connection
between this problem and the problem of finding the expected length
of the longest common subsequence of two random sequences.  
\end{abstract}

\section{Introduction}
There has recently been a great deal of work studying variations of
the following channel: $n$ bits are sent, but each bit is
independently deleted with fixed probability $p$.  This is the {\em
binary independently and identically distributed (i.i.d.) deletion
channel}, also called the {\em binary deletion channel} or just the
{\em deletion channel}.  To be clear, a deletion is different from an
erasure: if 10101010 was sent, the receiver would obtain 10011 if the
third, sixth, and eighth bits were deleted, but would obtain
$10?01?1?$ if the bits were erased.  The capacity of this channel
remains unknown.  See \cite{Mercier,MitzSurvey} for recent overviews.

In this work, we consider the zero-error capacity for the related
setting where up to $pn$ bits can be deleted.  For brevity, we refer
to the this variation as the {\em adversarial deletion channel}, as we
can imagine the channel as an adversary deleting bits in an attempt to
confuse the decoder, and the goal is choose a codebook that allows
successful decoding against this adversary.  Our main question of
interest is to consider what values of $p$ allow non-zero capacity in
this setting.  It is immediately clear, for example, that we require
$p < 0.5$, as an adversary that can delete half the bits can arrange
for any sent string to be received as either a string of all 0's or
all 1's.

Although this question appears quite natural, it does not seem to have
been specifically tackled in the literature.  It is certainly implicit in the early
work on insertion and deletion channels, most notably the seminal
works by Levenshtein\cite{Levenshtein} and Ullman \cite{Ullman},
although in most early papers the focus was on a constant number of
deletions or insertions, not a constant fraction.  The question of
codes that can handle a constant fraction of worst-case insertions or
deletions was specifically considered by Schulman and Zuckerman
\cite{schulman1999agc}, but they did not focus on optimizing the rate.

We consider two approaches that give provable lower bounds on what
values of $p$ allow non-zero capacity.  Our first bound is derived
from a simple combinatorial argument.  We provide it as a baseline,
and because it may be possible to generalize or improve the argument
in the future.  Our better bound arises from considering the expected
length of the longest common subsequence (LCS) of two random
sequences, a subject which has come under study previously. (See, for
example, the work of Lueker \cite{Lueker} and the references therein.)

Formally, subsequences are defined as follows: given a sequence $X =
(x_1,x_2,\ldots,x_n)$, the sequence $Z = (z_1,z_2,\ldots,z_m)$ is a
subsequence of $X$ if there is a strictly increasing set of indices
$i_1,i_2,\ldots,i_m$ such that
$$z_j = x_{i_j}.$$  Given two sequences $X$ and $Y$, the longest common
subsequence is simply a subsequence of $X$ and $Y$ of maximal length.  
The connection between
the adversarial deletion channel and longest common subsequences is
quite natural.  If two strings of length $n$ have an LCS of length at
least $(1-p)n$, then they cannot both be simultaneously in the
codebook for an adversarial deletion channel, as the adversary can
delete bits so that the received string is a subsequence common to both
strings.

It is known that the expected length of the LCS of two sequences of
length $n$ chosen uniformly at random converges to
$\gamma n$ for large $n$ and a constant $\gamma$,
as we explain more fully below.
Lueker shows $\gamma < 0.8269$ \cite{Lueker}.  Our main result is that
this implies that for any $p < 1 -\gamma$, we can achieve a non-zero
rate for the adversarial deletion channel.  Hence we can prove that
for $p < 0.1731$, a non-zero rate is possible.  In fact we generalize
our argument, considering the expected of the LCS of two sequences
chosen in an alternative, but still random, fashion.  While we do not
have provable bounds based on this generalization, experiments suggest
that the approach allows further significant improvements on the bounds, as we 
explain in Section~\ref{sec:fom}.

\section{A Simple Combinatorial Bound:  The Bipartite Model}

Our first approach is to work with a bipartite graph that
represents both the possible codewords and the possible received
strings.  For notational convenience, let us treat $pn$ as an integer,
and let us assume without loss of generality that the adversary
deletes a full $pn$ bits.  (The arguments for an adversary that delete
up to $pn$ bits differ only in $o(1)$ terms.)
For the rest of this section, our bipartite graph
will have a set $S$ of vertices, with one for each of the $2^n$
possible $n$-bit codewords, and a set $R$ of vertices, with one for
each of the the $2^{(1-p)n}$ possible received strings.  An edge
connects a vertex $v$ of $S$ to a vertex $w$ of 
$R$ if $w$ is a subsequence of $v$.  A valid codebook consists of a
subset ${\cal C}$ of $S$ with the property that every $w \in R$ is a
neighbor of at most one vertex in ${\cal C}$.
We derive a bound using this graph in terms of the binary entropy
function $H(p) = -p \log_2 p - (1-p) \log_2 (1-p)$.

\begin{theorem}
\label{thm1}
The capacity of the adversarial deletion channel is greater than zero
whenever $1-2H(p)+p > 0$.
\end{theorem}
\begin{proof}
A standard fact regarding subsequences is that each vertex of $R$
appears as a subsequence of
$$y = \sum_{j=0}^{pn} {n \choose j} = 2^{n(H(p)+o(1))}$$ vertices of
$S$ \cite{Chvatal}.  (The correct value can be seen by considering the vertex $R$ of all 0s;
the proof of the equality is a simple induction.)
Hence on average a vertex of $S$ is
adjacent to $2^{(1-p)n}y/2^n \approx 2^{n(H(p)-p+o(1))}$ vertices of $R$.  In
particular we now consider the set $S'$ of vertices that are adjacent
to at most twice the average of the vertices of $S$, and this set is
at least half the size of $S$.

We choose vertices of $S'$ sequentially to obtain our codebook in the
following manner.  We choose an arbitrary vertex from $S'$.  When we
choose a vertex of $S'$, it is adjacent to at most
$2^{n(H(p)-p+o(1))}$ vertices of $R$.  We remove from $S'$ the chosen
vertex, as well as any vertex whose neighborhood intersects that of the chosen
vertex.  This removes at most $y
2^{n(H(p)-p+o(1))}$ vertices from $S'$.  Continuing in this manner, 
we find a sequence of at least $|S'|/(y2^{n(H(p)-p+o(1))}) \geq 2^{n(1-2H(p)+p-o(1))}$ codewords.
In particular, we can achieve positive rate for $p$ that satisfies $1-2H(p)+p > 0$.
\end{proof}

This yields a positive rate for the adversarial deletion channel for
$p$ up to about $0.1334$. 

Before moving on, we make some remarks.  First, Theorem~\ref{thm1} can
be generalized in a straightforward manner to larger alphabets.  The
important step is that for alphabets of size $K$, each string of
length $j$ appears in $$\sum_{i=0}^{j} {n \choose i} (K-1)^{n-1}$$ supersequences of length
$n$. 

Second, because of its generality, one would hope that this argument
could be improved to yield significantly better bounds.  As an example,
we have considered a modified argument where we restrict our initial
set $S$ of possible codewords to consist of exactly $\alpha n$
alternating blocks of 0s and 1s, where $\alpha < 1-p$.  The idea is
that a smaller number of blocks, corresponding to a higher average
block length, might yield better results;  such arguments have proven
useful in the past \cite{diggavi2006ito,DM3}.  Our modification, however, did not 
appear to change our derived bound significantly (we did not perform
a full optimization after calculations showed the bound appeared to 
remain below $0.134$ with this choice of $S$).  

Obtaining improved results would appear to require a better
understanding of the structure of the bipartite graph: in particular,
greater insight into the degree distribution and neighborhood
structure of the vertices of $S$.  We circumvent this barrier
by moving to an alternate representation.

\section{A Bound from Longest Common Subsequences}

Our improved bound makes use of known results on longest common
subsequences of random strings.  As stated earlier, if two strings of
length $n$ have an LCS of length at least $(1-p)n$, then they cannot
both simultaneously be in the codebook for an adversarial deletion
channel.  This suggests a natural attack on the problem.  Consider
the graph where there is a vertex for each $n$-bit string.  The edges
of this graph will connect any pair of vertices that share a common
subsequence of length at least $(1-p)n$.  Then an independent set in
this graph corresponds to a valid codebook for the adversarial
deletion channel, in that no two codewords can be confused by only
$pn$ deletions.  We remark that this connection between codes and
independent sets is fairly common in the setting of deletion channels;
see, for example, \cite{sloane2002sdc}.

We require some basic facts, which we take from \cite{Lueker}; for an
early reference, see also \cite{Chvatal}.  Let $L(X,Y)$ denote the
length of the LCS of two strings $X$ and $Y$, and let $L_n$ be the
length of the LCS of two strings chosen uniformly at random.  By subadditivity
one can show that there exists a constant $\gamma > 0$ such that
$$\lim_{n \rightarrow \infty} \frac{\E[L_n]}{n} = \gamma.$$ The exact
value of $\gamma$ is unknown; \cite{Lueker} finds what appears to be
the best known bounds, of $0.788071 \leq \gamma \leq 0.826820$, and
computational experiments from \cite{Matzinger} suggest $\gamma
\approx 0.8182$.

The relationship between $\gamma$ and the threshold for zero-error decoding
is given by the following theorem.  
\begin{theorem}
\label{thm2}
The capacity of the adversarial deletion channel is greater than zero
whenever $p < 1 - \gamma$.  
\end{theorem}
\begin{proof}
For any $\gamma^* > \gamma$, for $n$ sufficiently large we have that
$L_n \leq \gamma^* n$.  Further, standard concentration inequalities
show that the length of the LCS of two strings $X$ and $Y$ chosen
uniformly at random is concentrated around its expectation $L_n$.
Such results are straightforward to derive; see, for example,
\cite{ArratiaWaterman}.  Formally, let $Z_i$ be the pair of bits in
position $i$ of the two strings.  The we have a standard Doob
martingale given by $A_i = {\bf E}[L(X,Y)~|~Z_1,Z_2,\ldots,Z_i]$.
The value of the bits $Z_i$ can change the value of $A_i$ by at most
2; one can simply remove these bits if they are part of the LCS, so
their value can only change the expected length of the LCS by at most 2.
Applying standard forms of Azuma's inequality (see, for example,
\cite{MitzenmacherUpfal}[Chapter 12] as a reference), we can conclude
$$\Pr \left [ L(X,Y) \geq L_n + \epsilon n \right ] \leq 2^{-f(\epsilon) n}$$
for some function $f$ of $\epsilon$ (that can depend on $q$).  

Now, let us focus on the graph described previously, where vertices correspond
to strings of length $n$, and edges connect vertices that share an LCS of length at least $\gamma' n$ for some $\gamma' > \gamma$.  It follows from the above that for sufficiently large $n$
the number of edges in the graph is at most 
$${2^n \choose 2}2^{-cn},$$
for some constant $c$.  
This is because the probability a randomly chosen edge is in the graph is at most $2^{-f(\epsilon) n}$ for 
an appropriate $\epsilon$.  (The probabilistic statement derived from Azuma's inequality would include edges
that correspond to self-loops, but our graph does not includes self-loop edges;  since we seek an upper bound
on the number of edges, this is not a problem.)

We can now apply a standard result on independent sets, namely Turan's theorem, which provides that  
in a graph with $j$ vertices and $k$ edges, there is an independent set of size at least $j^2/(2k+1)$
\cite{Turan}.  
In this context, Turan's theorem yields the existence
of an independent set, or codebook, of size at least 
$$\frac{2^{2n}}{2{2^n \choose 2}2^{-cn}+1} \geq 2^{cn-1}.$$  
That is, we have a code of positive rate whenever the fraction of deleted bits $p$ is at most $1-\gamma'$,
giving a non-zero rate for any $p > 1-\gamma$.
\end{proof}

This yields a provable positive rate for the adversarial deletion channel for
$p$ up to about $0.1731$, and an implied bound (assuming $\gamma \approx 0.8128$) of $0.1872$.

We note that while Theorem~\ref{thm1} uses a bipartite graph, it
should be clear that the process described greedily finds an
independent set on the corresponding graph used in Theorem~\ref{thm2},
where edges on the $2^n$ possible codewords correspond to paths of
length 2, or to shared subsequences, on the bipartite graph.
Theorem~\ref{thm1} results in weaker bounds, which is unsurprising,
given the sophisticated technical machinery required to obtain
the bounds on the expected LCS of random sequences.

Theorem~\ref{thm2} can also be generalized to larger alphabets, as the
techniques of \cite{Lueker} provide bounds on the corresponding
constant $\gamma_K$ for the expected LCS of strings chosen uniformly
at random from $K$-ary alphabets.

It might be tempting to hope that $1-\gamma$ is a
tight threshold.  Consider the graph obtained when $p$ is slightly
larger than $1-\gamma$; this graph will be almost complete, as almost
all edges will exist, and therefore it might seem that there might be
no independent set of exponential size to find. However, it might be possible
to obtain a large subgraph that is less dense, which would in turn yield 
a better threshold.  For
example, the string consisting of all 0s has very small degree
(comparatively) in the graph; by finding a suitable set of ``small
degree'' vertices, one might hope to increase the threshold for
non-zero capacity.

In the next section, we use this insight to generalize our argument above
in order to obtain, at least empirically, improved bounds. 

\section{Using First-Order Markov Chains}
\label{sec:fom}

To generalize the argument of Theorem \ref{thm2}, we consider random strings generated by
symmetric first-order Markov chains.  That is, the initial bit in the string is
0 or 1 with probability $1/2$; thereafter, subsequent bits are
generated independently so that each matches the previous bit with
probability $q$, and changes from $b$ to $1-b$ with probability $1-q$.
When $q = 1/2$, we have the standard model of randomly generated
strings.  Below we consider the case where $q > 1/2$.  Again, first-order Markov chains have proven useful in the
past in the study of deletions channels \cite{diggavi2006ito,DM3}, so
their use here is not surprising. 

 Subadditivity again gives us that
under this model, the expected longest common subsequence of two
randomly generated string again converges to $\gamma_q n$ for some
constant $\gamma_q$.  We claim the following theorem.
\begin{theorem}
\label{thm3}
For any constant $q$, the capacity of the adversarial deletion channel is greater than zero
whenever $p < 1 - \gamma_q$.  
\end{theorem}
\begin{proof}
Following the proof of Theorem~\ref{thm2}, let $L_n$ be the expected
length of the LCS of two strings
$X$ and $Y$ of length $n$ chosen according to the above process.  

For any $\gamma_q^* > \gamma_q$, for $n$ sufficiently large we have
that $L_n \leq \gamma^* n$.  As before, we can make use of
concentration, although the use of Azuma's inequality is a bit
more subtle here.  Again, let $Z_i$ be the pair of bits in position
$i$ of the two strings, so that we have a standard Doob martingale
given by $A_i = {\bf E}[L(X,Y)~|~Z_1,Z_2,\ldots,Z_i]$.  The value of
the bits $Z_i$ are important, in that the expected value of the LCS
conditioned on these bits depends non-trivially on whether the bits
are the same or different.  However, we note that if the bits do not
match, then after an expected constant number of bits (with the
constant depending on $q$), the bits will match again, and hence the
value of the bits $Z_i$ still can change the value of $A_i$ by at most
a constant amount.  So as before we, we can conclude
$$\Pr \left [ L(X,Y) \geq L_n + \epsilon n \right ] \leq 2^{-f(\epsilon) n}$$
for some function $f$ of $\epsilon$.  

We now use a probabilistic variation of Turan's theorem to show the
existence of suitably large independent sets in the graph where 
vertices correspond to strings of length $n$, and edges connect
vertices that share an LCS of length at least $\gamma' n$ for some
$\gamma' > \gamma_q$.  Our proof follows a standard argument based on
the probabilistic method (see, e.g., \cite{MitzenmacherUpfal}[Theorem
6.5]).  For each sequence $s$ of length $n$, let $\Delta(s)$ be the
number of times the sequence changes from 0 to 1 or 1 to 0 when read
left to right.  Then the probability of $s$ being generated by our first
order Markov chain process is 
$$p(s) = \frac{1}{2}q^{n-1-\Delta(s)}(1-q)^{\Delta(s)}.$$
Consider the following randomized algorithm for generating an independent
set:  delete each vertex in the graph (and its incident edges) independently
with probability $1 - 2^n z p(s)$ for a $z$ to be determined later, and then for
each remaining edge, remove it and one of its adjacent vertices.

If $Q$ is the number of vertices that survive the first step, we have
$$\E[Q] = \sum_{s \in \left \{ 0,1 \right \}^n} 2^n z p(s) = z2^n.$$
Now let $R$ be the number of edges that survive the first step. Then
$$\E[R] \leq \sum_{s,t \in \left \{ 0,1 \right \}^n} z^2 2^{2n} p(s)p(t){\bf 1}(L(s,t) > \gamma' n).$$
Notice that $\E[R]$ is bounded by $z^2 2^{2n} \cdot \Pr[ L(X,Y) \geq \gamma' n ]$
for a random $X$ and $Y$ chosen according to the first order Markov process.  
As argued previously, we therefore have $\E[R] \leq z^2 2^{(2-c)n}$ for some constant $c$.  
Choosing $z = 2^{cn-n-1}$ gives an expected independent set of size at least
$\E[Q - R] \geq 2^{cn-2}$, showing that there must exist an independent set which gives
a non-zero rate for any $p < 1 - \gamma_q$.  

It remains to check that in fact our choice of $z$ ensures that 
$1 - 2^n z p(s)$ is indeed a valid probability;  that is, we need
$2^n z p(s) \leq 1$ for all $s$.   We have 
$2^n z p(s) = 2^{cn-1} p(s)$;  in our argument we can in fact choose $c$ to be a
sufficiently small constant (specifically, choose $c < \log_2 (1/q)$) to ensure
that this is a valid probability.
\end{proof}

While we do not have a formal, provable bound on $\gamma_q$ for $q > 1/2$, 
such bounds for specific $q$ might be obtained using the methods of
\cite{Lueker}.  We leave this for future work.  Besides requiring
generalizing these methods to strings chosen non-uniformly, we note
that another difficulty is that these techniques require large-scale 
computations where the tightness of the bound that can be feasibly obtained
may depend on $q$.  

Empirically, however, we can simply run our first-order Markov chains
to generate large random strings and compute the LCS to estimate the
behavior as $q \rightarrow 1$.  (For more detailed estimation methods
for LCS problems, see for example \cite{Matzinger}.)  Our experiments
suggest that the LCS of two sequences randomly generated in this
manner may converge to $0.75n$ or something near that quantity.  For
example, we found the length of the LCS for 1000 pairs of sequences of
length $100000$ with $q = 0.95$; the average length was $77899.4$,
with a minimum of $77499$ and a maximum of $78375$.  For 1000 pairs of
length $100000$ with $q = 0.99$, the average length of the LCS was
$77479.8$, with a minimum of $76083$ and a maximum of $78831$.  For
1000 pairs of length $100000$ with $q = 0.999$, the average length of
the LCS was $75573.2$, with a minimum of $68684$ and a maximum of
$81483$.  Naturally the variance is higher as runs of the same bit
increase in length;  however, the trend seems readily apparent.

It is also worth noting that Theorem~\ref{thm3} can generalize
beyond choosing vertices according to a distribution given by
first-order Markov chains.  The key aspects -- subadditivity providing
a limiting constant, concentration using Azuma's inequality, and some
reasonable variation of Turan's theorem based on the probabilistic
method -- might naturally be applied to other distributions over
$n$-bit strings as well.  Indeed, we emphasize that even if it can be
shown that $\gamma_q$ approaches $0.75$ as $q$ approaches 1, it may
still be that this does not bound the full range of values of $p$ for
which the adversarial deletion channel has nonzero capacity.

\section{Conclusion}

We have given a proof that for any $p < 1 - \gamma_q$ the rate of the
adversarial deletion channel is positive, where $\gamma_q$ is the
value such that two strings of length $n$ chosen according to a
symmetric first-order Markov chain that flips bits with probability
$1-q$ has, in the limit as $n$ goes to infinity, an LCS of expected
length $\gamma_q n$.  For $q = 1/2$, bounds on $\gamma_q$ are known.

Many natural open questions arise from our work.  Can we find tighter
bounds for $\gamma$? Can we derive a stronger bound using $\gamma_q$?
Are there better approaches for choosing codebooks in this setting?
What more can we say about the structure of the graphs underlying our
arguments?  All of these questions relate to our main question here:
can we improve the bounds on $p$ for which the adversarial
deletion channel has non-zero rate?

Our theorems also implicitly give bounds on the rate (or equivalently
the size of the codebook) obtainable away from the threshold.  We have
not attempted to optimize these bounds, as we expect our arguments are
far from tight in this respect.  For example, we note that in
Theorem~\ref{thm2}, the bound on the size of the independent set
obtained can seemingly be strengthened; specifically, the
probabilistic method gives a lower bound on the size of the
independent set of a graph with vertex set $V$ and vertex degrees
$d(u)$ for $u \in V$ of $\sum_{u \in V} \frac{1}{d(u)+1}$
\cite{spencer1994ten}.  Such bounds seem difficult to use in this
context, and would would not appear to change the threshold obtained
with this proof method without additional insights.  Finding good
bounds on the rate for this channel, and of course finding good explicit
codes, remain open questions.  

Finally, we mention that for a large class of adversarial
channels known as {\em causal adversary} or {\em memoryless channels},
it is known that no positive rate is achievable for $p>0.25$ \cite{GS10, LJD09}.
Proving an equivalent upper bound for the adversarial deletion channel
is a tantalizing goal, especially in light of the result from Section
\ref{sec:fom} suggesting that for $p$ close to $0.25$, positive rate
{\em is} achievable.  

Unfortunately, it is not clear how to extend the techniques
of \cite{GS10, LJD09} to the adversarial deletion channel. At a high level, the results
of \cite{GS10, LJD09} rely on the fact that, given two codewords
$\mathbf{x}$ and $\mathbf{x'}$ with Hamming distance $d$, an adversary
can cause a ``collision" between $\mathbf{x}$ and $\mathbf{x'}$ by
{\em flipping} $d/2$ bits of each. So as long as the adversary can
find two codewords of Hamming distance at most $n/2$, just $n/4$ bit
flips are needed to cause a collision. 

In contrast, it is possible for two $n$-bit strings $\mathbf{x}$ and
$\mathbf{x'}$ of Hamming distance $d$ to have no common subsequence of
length $n-d+1$. That is, it might require $d$ deletions to both
$\mathbf{x}$ and $\mathbf{x'}$ to cause a collision, while it would
only require $d/2$ bit-flips. 
This appears to be the fundamental difficulty in proving that for any
specific $p<0.5$, no 
positive rate is achievable for the adversarial deletion
channel. Proving such an upper bound 
is perhaps our biggest open question.

\section*{Acknowledgments}

Ian Kash is supported by a postdoctoral fellowship from the Center for
Research on Computation and Society.
Michael Mitzenmacher performed part of this work while visiting
Microsoft Research New England, and thanks Madhu Sudan and Venkat
Guruswami for early discussions on this problem.  His work is also
supported by NSF grants CNS-0721491, CCF-0915922, and IIS-0964473.
Justin Thaler is supported by the 
Department of Defense (DoD) through the National Defense Science \&
Engineering Graduate Fellowship (NDSEG)  
Program.  Jonathan Ullman is supported by NSF grant CNS-0831289.
The authors also thank Varun Kanade, Tal Moran, and Thomas Steinke for helpful discussions.

\bibliographystyle{plain}

\end{document}